\documentclass[12pt]{article}
\usepackage{amsmath,amssymb,amsfonts,amsthm,amscd,yhmath}
\usepackage[colorlinks=true,urlcolor=blue,citecolor=blue,linkcolor=blue,pdfstartview=FitH,pdfpagemode=None]{hyperref}

\hypersetup{
    pdftitle={Variational Principles for Water Waves}, 
    pdfauthor={Boris Kolev and David Sattinger}, 
    baseurl={}, 
    pdfsubject={}, 
    pdfkeywords={}, 
    }

\newtheorem{theorem}{Theorem}[section]

\newtheorem{lemma}[theorem]{Lemma}
\newtheorem{proposition}[theorem]{Proposition}

\theoremstyle{definition}

\newtheorem*{remark}{Remark}

\numberwithin{equation}{section}

\newcommand{\pd}[2]{\frac{\partial{#1}}{\partial{#2}}}

\newcommand{\pdl}[2]{\frac{\delta{#1}}{\delta{#2}}}

\newcommand{\set}[1]{\left\{#1\right\}}
\newcommand{\inmp}{\int_{-\infty}^{+\infty}}

\newcommand{\R}{\mathbb R}

\newcommand{\DD}{\mathcal{D}}
\newcommand{\EE}{\mathcal{E}}
\newcommand{\FF}{\mathcal{F}}
\newcommand{\HH}{\mathcal{H}}
\newcommand{\KK}{\mathcal{K}}
\newcommand{\LL}{\mathcal{L}}
\newcommand{\MM}{\mathcal{M}}

\newcommand{\bfv}{\mathbf{v}}
\newcommand{\bfw}{\mathbf{w}}

\newcommand{\bfN}{\mathbf{N}}
\newcommand{\bfx}{\mathbf{x}}

\newcommand{\bfX}{\mathbf{X}}

\newcommand{\wt}{\widetilde}

\DeclareMathOperator{\grad}{\nabla} %
\DeclareMathOperator{\curl}{curl} %
\DeclareMathOperator{\divg}{div} %

\begin{document}

\title{Variational Principles for Water Waves}%
\author{B. Kolev \\C.M.I., Universit\'{e} de Provence \\Marseille, France
    \and D. H. Sattinger \\ Yale University \\New Haven, USA}

\date{March 8, 2006}%

\maketitle


\begin{abstract}
We describe the Hamiltonian structures, including the Poisson
brackets and Hamiltonians, for  free boundary problems for
incompressible fluid flows with vorticity. The Hamiltonian structure
is used to obtain variational principles for stationary gravity
waves both for irrotational flows as well as flows with vorticity.
\end{abstract}


\section{Introduction}

In 1933 Friedrichs \cite{Fri33} proposed the functional
\begin{equation*}
    J(\psi) = \iiint\limits_{0\le\psi\le 1}\left[(\grad\psi)^2+v^2(x,y)\right]\,d^2\bfx,
\end{equation*}
where $\psi$ is the stream function for an incompressible flow, as a
variational method of obtaining solutions to free boundary value
problems. Critical points of $J$ are harmonic functions which
satisfy the condition
\begin{equation*}
    (\grad\psi)^2=v^2
\end{equation*}
on the free boundary, given by $\psi=1$. The free boundary condition
relevant to theory of gravity waves, however, is the Bernoulli
equation
\begin{equation*}
    \frac{(\grad\varphi)^2}{2}+g\zeta = constant,
\end{equation*}
where $\varphi$ is either the velocity potential for irrotational
flow, or the stream function in the case of flows with vorticity.
Thus some other variational principle is needed for the study of
gravity waves.

Recently, a variational principle for gravity waves with vorticity
was given by Constantin et. al. \cite{ASS06}, using a direct, "hands
on" approach. More generally, a variational principle for a
stationary wave may be obtained for systems possessing a Hamiltonian
structure by minimizing the Hamiltonian computed in a Galilean frame
moving with the wave. We illustrate that approach in this study.

We begin with a brief review of Euler's equations of incompressible
flows, and the associated free boundary value problems; in
\S\ref{PS} we describe the Hamiltonian structure of these problems,
for irrotational flows and flows with vorticity, as given by Lewis
et.al. \cite{LMMR86}. All the functions under consideration in this
article, including the free boundaries, are assumed to be smooth.


\section{Incompressible fluid flows}\label{iff}

Let the velocity field of an incompressible fluid in a fixed region
$\DD$ be denoted by ${\bf v}$. The incompressibility of the fluid is
expressed by the condition $\divg\,\bfv=0.$ We must have $\bfv \cdot
\nu=0$ on the boundary of $\DD$, where $\nu$ is the outward unit
normal at the boundary. Euler's equations of motion for the flow of
an inviscid, incompressible fluid are
\begin{gather}
\frac{d\bfx}{dt}=\bfv, \qquad
\rho\frac{d\bfv}{dt}=\rho(\bfv_t+(\bfv\cdot\grad)\bfv)=
-\grad(\,p-{\bf g}\cdot\bfx),\label{euler1}\\[4mm] \divg\,\bfv=0,\nonumber
\end{gather}
where $\rho$ is the density, $p$ is the hydrodynamic pressure, and
${\bf g}\cdot \bfx$ is the gravitational potential. Henceforth we
take $\rho=1$.

Given a manifold $\DD\in\R^3$ with smooth boundary, we denote by
$\LL^2(\DD)$ the Hilbert space of vector fields on $\DD$ with the
inner product
\begin{equation*}
    \langle \bfv,\bfw\rangle=\iiint_\DD \bfv\cdot\bfw\,d^3\bfx.
\end{equation*}
We denote by $\LL_\pi$ the closed subspace of $\LL^2(\DD)$ generated
by vector fields; of the form $\bfw=\grad\,p$ for some function $p$
with finite Dirichlet norm. The orthogonal complement
$\LL_\sigma=\LL_\pi^\perp$ is the space of all vector fields $\bfv$
for which $\langle \bfv,\grad\,p\rangle=0$ for all $p\in
W^{1,2}(\DD)$.  By applying the Gauss divergence theorem, we see
that if $\bfv\in\LL^2(\DD)$ and is smooth, say $C^1$, then
$\divg\,\bfv=0$ and $\bfv\cdot\nu=0$ on $\Sigma=\partial\DD$, where
$\nu$ denotes the outward unit normal on $\Sigma$. The Hilbert space
$\LL_\sigma$ is the space of weakly divergence-free vector fields.
We denote the orthogonal projections onto $\LL_\sigma$ and $\LL_\pi$
by $P_\sigma$ and $P_\pi$ respectively.

In many applications the fluid is not confined to a fixed region,
but instead carries the region with it. In such cases, the region
$\DD$ occupied by the fluid must also be determined. Such problems
are called free boundary problems and occupy a substantial part of
the literature on incompressible flows.

Given an irrotational flow ($\curl\,\bfv=0$) on a simply connected
domain, there is velocity potential $\varphi$ for which
$\bfv=\grad\varphi.$ The velocity potential is defined only up to an
arbitrary function of time; the transformation $\varphi\mapsto
\varphi+k(t)$ is called a \emph{gauge transformation}, and will play
a role in what follows.

The equation $\divg\,\bfv=0$ implies that $\varphi$ is harmonic.
Substituting $\bfv=\grad\varphi$ into the second equation in
\eqref{euler1} we obtain
\begin{equation*}
    \grad\left(\varphi_t+\frac12(\grad\varphi)^2 + gz + p\right) = 0,
\end{equation*}
hence
\begin{equation*}
    \varphi_t+\frac12(\grad\varphi)^2 + gz + p = k(t)
\end{equation*}
for some function of time, which can be eliminated by a gauge
transformation of the velocity potential. We always choose the gauge
to be such that
\begin{equation*}
    \varphi_t+\frac12(\grad\varphi)^2 + gz + p=0
\end{equation*} everywhere in the
fluid.

An interface between the fluid and another medium, for example air,
is called a \emph{free surface}. If the pressure is constant in the
air, then it is also constant at the surface of the fluid, and we
may normalize the pressure to be zero at the free surface. Hence we
obtain Bernoulli's equation
\begin{equation*}
    \varphi_t+\frac12(\grad\varphi)^2+gz=0,
\end{equation*}
where $gz$ is the gravitational potential on the free surface.

The free surface is given in space-time by $\phi=0$, where
$\phi(x,y,z,t)=z-\zeta(x,y,t)$. The free surface moves with the
fluid, hence the material derivative of $\phi$ vanishes, and
\begin{equation*}
    0=\frac{d\phi}{dt}=\frac{d}{dt}(z-\zeta)=v^3-\zeta_t-v^1\zeta_x-v^2\zeta_y;
\end{equation*}
hence
\begin{equation*}
    \zeta_t+v^1\zeta_x+v^2\zeta_y-v^3=0.
\end{equation*}
This is called the \emph{kinematic} condition on the free surface.

This collection of equations for gravity waves on a free surface is
known as Euler's equations for waves on the surface of an inviscid,
incompressible fluid with irrotational flow in the region
$\DD=\{(x,y,z)\,:\,0\le z \le h+\zeta (x,y,t)\}$. They are
\begin{alignat}{3}
\Delta \varphi &=0   \qquad & 0\le y\le &h+ \zeta, \nonumber\\[4mm]
\zeta_t+\varphi_x \zeta_x+\varphi_y\zeta_y &=\varphi_z  \qquad &
\mbox{on} \ S ; \label{kin}\\[4mm]
\varphi_t+\frac{1}{2}|\nabla \varphi |^2+g\,z &= 0 \qquad
&\mbox{on}\ S;\label{bern}\\[4mm]
\varphi_z&=0 \qquad  &\mbox{on} \ z=&0. \nonumber
\end{alignat}

Here, $\varphi$ is the velocity potential of the flow, and
$\zeta(x,y,t)$ the displacement of the fluid surface from
equilibrium. We have neglected surface tension. The second equation
is known as the {\it kinematic equation}; the third equation is
Bernoulli's equation. At rest, the fluid lies in the region $0\le z
\le h$; $g$ is the acceleration due to gravity. The free surface is
denoted by $S=\{(x,y,z):\,z=h+\zeta(x,y,t)\}.$

The two physical constants in the theory are $g$ and $h$. Let $c$
denote a characteristic velocity (e.g. the velocity of a gravity
wave); then $h/c$ is a characteristic time. We introduce
dimensionless variables
\begin{equation*}
(x,y,z)=h(x',y',z'), \qquad t=ht'/c, \qquad \varphi=ch\varphi'.
\end{equation*}
The equation \eqref{bern} now becomes
\begin{equation}
\varphi^\prime_{t'}+\frac12 (\grad^\prime\varphi^\prime)^2+
\lambda\,\zeta=0, \qquad \lambda=\frac{gh}{c^2},\label{bern2}
\end{equation}
where $\lambda$ is the inverse square of the {\it Froude} number.
The other equations in Euler's system are unchanged under the
rescaling. From now on we drop the primes and understand that we are
working in non-dimensional variables.

The Euler equations are invariant under the one-parameter subgroup
of Galilean boosts along the $x$ axis, given by
\begin{equation*}
(x',y',t')=(x-ct,y,t), \qquad \bfv'(x',y',t')=\bfv(x,y,t)-(c,0).
\end{equation*}
The velocity potential, however, is determined only up to a function
of time. Thus the Galilean boosts on the velocity potential are
given by  $$\varphi'(x',y',t')=\varphi(x,y,t)-cx+q(t).$$ Under these
Galilean boosts,
\begin{equation}\label{PotGal}
\pd{\varphi' }{t'}+\frac12(\grad'\varphi')^2=\pd{\varphi
}{t}+\frac12(\grad\varphi)^2 +q'(t)-\frac12 c^2.
\end{equation}
The result follows by direct calculation, noting that
$$\frac{\partial }{\partial t'}= \frac{\partial }{\partial
t}+c\frac{\partial }{\partial x}, \qquad \frac{\partial }{\partial
x'}=\frac{\partial }{\partial x}.$$

\begin{proposition}
Suppose the solutions of Euler's equations are
 stationary in a Galilean frame moving with speed $c$. Then $\zeta_{t'}=\varphi'_{t'}=0$; and, choosing $q(t)=c^2t$,  the conditions on the free surface are (dropping the primes)
\begin{equation}\label{stat_Bernoulli}
\varphi_x \zeta_x =\varphi_y ,\qquad
\frac{\varphi_x^2+\varphi_y^2}{2} +\lambda\,\zeta = \frac{c^2}2.
\end{equation}
\end{proposition}

\begin{proof}
By \eqref{PotGal} the Bernoulli equation in the moving frame is
\begin{equation*}
\varphi'_{t'}+\frac12(\grad'\varphi')^2+\lambda\zeta'=q'(t')-\frac12
c^2.
\end{equation*}
As $x\to\pm\infty$ $\zeta'\to 0$ while $(\grad'\varphi')^2\to c^2$.
Moreover, $\varphi'_{t'}=0$ by the assumption of stationarity. These
conditions force the choice $q'=c^2,$ and the result follows. The
kinematic equation in the moving frame is immediate.
\end{proof}

\begin{proposition}
Let $\bfv$ be a divergence-free vector field in a domain $\DD$.
There is a unique orthogonal decomposition, known as the Weyl-Hodge
decomposition,
\begin{gather}
\bfv=\bfw+\grad\varphi, \label{wh} \\[4mm]
\Delta\varphi=0,\quad \varphi_\nu=\bfv\cdot\nu; \qquad
\divg\,\bfw=0, \quad \bfw\cdot\nu=0.
\end{gather}
\end{proposition}

The proof is left to the reader.


\section{Poisson structures}\label{PS}

Let $M$ be a $C^\infty$ manifold of dimension $n$, and let $F,G\in
C^\infty(M)$. A bilinear form  $\{F,G\}$ is said to be a {\it
Poisson bracket} if
\begin{itemize}
\item $\{F,G\}=-\{G,F\}$; \item $\{F,GH\}=\{F,G\}H+G\{F,H\}.$
\item $\{\{F,G\},H\}+\{\{G,H\},F\}+\{\{H,F\},G\}=0;$
\end{itemize}
The second property implies that the Poisson bracket is a derivation
in each of its entries. Hence any $H\in C^\infty(M)$ generates a
vector field $X_H$, called a {\it Hamiltonian} vector field on $M$,
defined by $X_HF=\{H,F\}$. The Hamiltonian vector field $X_H$
generates a flow on $M$; if $x^i$ are a set of local coordinates on
$M$, then the time evolution of the $x^i$ on that chart is given by
the ordinary differential equations
\begin{equation*}
    \dot x^i=\{H,x^i\}.
\end{equation*}

Due to the fact that the bracket acts as a derivation on each of its
entries, we may represent a Poisson bracket in the form
\begin{equation*}
    \{F,G\}=\sum_{j,k=1}^n W^{jk}\pd{F}{x^j}\pd{G}{x^k},
\end{equation*}
where $W^{jk}(x)$ is a skew-symmetric matrix.

If $\det W\ne 0$ then it is easily seen that $n$ must be even. A
classical theorem of Darboux states that in this case it is always
possible to find a set of local coordinates, called {\it canonical}
coordinates $q^i,p^i$, ($1\le i\le n/2$) in which the Poisson
brackets take the form
\begin{equation*}
    \{F,G\}=\sum_{j=1}^n \pd{F}{p^j}\pd{G}{q^j}-\pd{F}{q^j}\pd{G}{p^j}.
\end{equation*}

A manifold with a Poisson bracket is called a \emph{Poisson
manifold}; if the brackets are non-degenerate, the manifold is
called a {\it symplectic} manifold. On a symplectic manifold, the
Hamiltonian flow takes the form
\begin{equation*}
\dot q^i=\pd{H}{p^i}, \qquad \dot p^i=-\pd{H}{q^i}.
\end{equation*}

In this paper we shall restrict ourselves to the case in which $M$
is a linear vector space with an inner product $\langle \ , \
\rangle$; and we shall write the Poisson brackets in the form
\begin{equation*}
    \{F,G\}=\langle \grad F,J_x\grad G\rangle,
\end{equation*}
where $J_x$ is a skew-symmetric linear transformation on $M$ and
$\grad F$ is the gradient of the function $F$. The gradient is
characterized as follows. Differentiating $F(x(t))$ along a curve
$x(t)$ on $M$, we have
\begin{equation*}
    \frac{d}{dt}F=\langle \grad F, \dot x\rangle.
\end{equation*}

If $J_x$ is non-singular, then the Poisson brackets are
non-degenerate and have locally a canonical system of coordinates.
In many problems of physical interest, however, the Poisson brackets
are degenerate, i.e. $\det J_x=0.$ For example, in the study of
rigid motions about a fixed point in $\R^3$, the Poisson bracket is
\begin{equation}\label{rotations}
    \{F,G\}=\langle \grad F,\, {\bf x}\times\grad G\rangle.
\end{equation}
The operator $J_\bfx$ is defined by $J_\bfx \bfv=\bfx\times\bfv;$
hence $\ker(J_\bfx)=\R\bfx.$

The bracket \eqref{rotations} vanishes for all regular functions $G$
whenever $F$ is spherically symmetric. Such a function $F$ is called
a \emph{Casimir}. It is  invariant under any Hamiltonian flow
generated by these brackets.

Any Poisson bracket on an odd-dimensional manifold must be
degenerate and therefore have Casimirs. The bracket
\eqref{rotations} is an example of a non-canonical Poisson bracket.

The formalism of Poisson brackets and Hamiltonian flows can be
extended to infinite dimensions, for example, in the study of
continuum mechanics, though a number of technical difficulties
arise. In particular, Poisson structures play a useful role in the
theory of the Euler equations for an incompressible fluid. Two
important such brackets are the Poisson bracket introduced by Arnold
\cite{Arn66,Arn69,AK98} in his study of incompressible fluids on
fixed domains, and the Poisson bracket implicit in Zakharov's
 fundamental discovery \cite{Zak68}  of the Hamiltonian
structure of the Euler equations of gravity waves.

\subsection{Arnold's Poisson brackets}

Arnold observed that Euler's equations for an incompressible fluid
in a fixed domain $\DD$ are directly analogous to his equations for
rigid body motion, and that they have a Hamiltonian structure with
the Hamiltonian and Poisson brackets given respectively by
\begin{equation}\label{H0}
    H=\iiint\limits_\DD \frac12 \bfv\cdot\bfv\,d^3\bfx,
\end{equation}
and
\begin{equation}\label{P0}
\{F,G\,\}=\iiint\limits_\DD \pdl{ F}{ \bfv}\cdot
\left(\curl\,\bfv\times\pdl{ G}{\bfv}\right)\,d^3\bfx.
\end{equation}
Here, $F$ and $G$ are functionals on $\LL_\sigma$ with gradients in
$\LL_\sigma$. The gradient of $F$ is $\delta F/\delta \bfv$, the
Euler-Lagrange derivative of $F$ with respect to $\bfv$.  For
example, $\pdl{H}{\bfv}=\bfv.$ The operator $J_\bfv$ in this case is
\begin{equation*}
    J_\bfv\bfw=P_\sigma (\curl\,\bfv\times\bfw).
\end{equation*}

Let us show that \eqref{euler1} are the Hamiltonian equations
generated by \eqref{H0} and \eqref{P0}. We have
\begin{gather*}
\dot F=\iiint\limits_\DD\pdl{F}{\bfv}\cdot
\bfv_t\,d^3\bfx, \\[4mm] \{H,F\}=\iiint\limits_\DD
\pdl{H}{\bfv}\cdot (\curl\,\bfv\times \pdl{F}{\bfv})\,
d^3\,\bfx=\iiint\limits_\DD \pdl{F}{\bfv}\cdot(\bfv\times
\curl\,\bfv)\, d^3\,\bfx.
\end{gather*}
The Hamiltonian flow $\dot F=\{H,F\}$ implies that
\begin{equation}\label{allF}
\iiint\limits_\DD\pdl{F}{\bfv}\cdot (\bfv_t+(\curl\,\bfv)\times\bfv)
\,d^3\bfx=0
\end{equation}
for all admissible $F$ on $\LL_\sigma$.

All linear functionals of the form $F_\bfw(\bfv)=\langle
\bfw,\bfv\rangle$ are admissible, and the gradient of $F_\bfw$ is
the vector $\bfw$. Therefore $\bfv_t+(\curl\,\bfv)\times\bfv $
belongs to $\LL_\sigma^\perp=\LL_\pi$. Hence it is a gradient, and
\begin{equation*}
    \bfv_t+(\curl\,\bfv)\times\bfv = \grad f
\end{equation*}
for some function $f$. The Euler momentum equations \eqref{euler1}
follow from this and the vector identity
\begin{equation}
\label{vecid}(\curl\,\bfv)\times \bfv=(\bfv\cdot\grad)\bfv-\frac12
\grad |\bfv|^2,
\end{equation}
if we let $f = - p + \frac12 \grad |\bfv|^2$.

Just as in the case of rigid motion, the Arnold bracket is
degenerate. This degeneracy is related to the action of the (formal)
group of volume preserving diffeomorphisms acting on $\DD$. Arnold's
Poisson bracket is an example of a \emph{Lie-Poisson} bracket.

\subsection{Zakharov's Poisson brackets}\label{zak}

In 1968, Zakharov made a striking observation: Euler's equations for
\emph{irrotational gravity waves} have a canonical Hamiltonian
structure. The Hamiltonian (in non-dimensional variables) is
\begin{equation*}
H=\frac12\iiint\limits_\DD(\grad\varphi)^2\,d^3\bfx+
\frac12\lambda\,\iint\limits_{\R^2}\zeta^2(x,y,t)\,d^2\bfx.
\end{equation*}
The Poisson brackets implicit in Zakharov's observation are the
canonical brackets
\begin{equation*}
\{F,G\}=\iint\limits_{\R^2} \left(\pdl{ F}{\varphi}\pdl{
G}{\zeta}-\pdl{ F}{\zeta}\pdl{ G}{\varphi}\right)d^2{\bf x};
\end{equation*}
the Hamiltonian flow is then the canonical flow
\begin{equation*}
    \zeta_t=\pdl{ H}{\varphi}, \qquad \varphi_t=-\pdl{ H}{\zeta}.
\end{equation*}

The Hamiltonian $H$ is regarded as a functional of
$(\wt\varphi,\,\zeta)$ where $\zeta=\zeta(x,y,t)$ is the height of
the free surface, and $\wt\varphi=\varphi\vert_S$ is the trace of
the harmonic function $\varphi$ on the free surface, with
$\varphi_\nu=0$ on the bottom. The evolution takes place in the
space of harmonic functions on $\DD$.

Zakharov's result is verified by calculating the gradients of $H$
with respect to $\zeta$ and $\varphi$. Now
\begin{equation*}
\frac{d}{d\varepsilon}H(\varphi,\zeta_\varepsilon)\Big\vert_{\varepsilon=0}=\iint\limits_{\R^2}
\left[\frac12(\grad\wt\varphi)^2+\lambda\zeta\right]\delta\zeta\,d^2\bfx,
\end{equation*}
where $\grad\wt\varphi$ denotes $\grad\varphi\big\vert_{\R^2}.$ By
identification,
\begin{equation*}
\frac{\delta
H}{\delta\zeta}=\frac12(\grad\wt\varphi)^2+\lambda\zeta.
\end{equation*}

Similarly,
\begin{multline*}
\frac{d}{d\varepsilon}H(\wt\varphi_\varepsilon,\zeta)\Big\vert_{\varepsilon=0}=
\iiint\limits_\DD\grad\varphi\cdot\grad\delta\varphi\,d^3\bfx\\[4mm]
=-\iiint\limits_\DD\delta\varphi\Delta\varphi\,d^3\bfx+
\iint\limits_\Sigma\delta\varphi\pd{\varphi}{\nu}dS=
\iint\limits_\Sigma\delta\varphi\pd{\varphi}{\nu}dS,
\end{multline*}
since $\varphi$ is harmonic in $\DD$ and $\varphi_\nu=0$ on the
bottom.

On the free surface
\begin{equation*}
\wt\varphi_\nu dS=\grad\wt\varphi \cdot
\frac{(-\zeta_x,-\zeta_y,1)}{\sqrt{1+\zeta_x^2+\zeta_y^2}}\sqrt{1+\zeta_x^2+\zeta_y^2}d^2\bfx;
\end{equation*}
so
\begin{equation*}
\frac{\delta
H}{\delta\wt\varphi}=\wt\varphi_z-\wt\varphi_x\zeta_x-\wt\varphi_y\zeta_y.
\end{equation*}
The free boundary equations \eqref{kin} and \eqref{bern} are thus
precisely the Hamiltonian equations for this system.

\begin{remark}
The effects of surface tension can be obtained by simply adding the
boundary integral
\begin{equation*}
    \sigma\iint\limits_{\Sigma} dS
\end{equation*}
to the Hamiltonian, where $\sigma$ is the coefficient of surface
tension, and $dS$ is the element of surface area on the free surface
$S$. The inclusion of surface tension leads to an additional term in
the Bernoulli equation; when the free surface is a graph
$z=\zeta(x,y,t)$, it is
\begin{equation*}
\varphi_t+\frac{1}{2}|\nabla \varphi |^2+g\,z =
\sigma\,\divg\frac{\grad\zeta}{\sqrt{1+(\grad\zeta)^2}}, \qquad
\grad\zeta=(\zeta_x,\zeta_y).
\end{equation*}
The potential energy can also be written as the integral of the
gravitational potential over the fluid domain, so that the
Hamiltonian for gravity waves including the effects of surface
tension is
\begin{equation}\label{totHam}
H=\iiint\limits_\DD\left[ \frac{(\grad\varphi)^2}2+\lambda\,
U_+(\bfx)\right]\,d^3\bfx+\sigma\iint\limits_SdS,
\end{equation}
where $U_+(\bfx)$ is the gravitational potential, truncated in such
a way that the integral over the unbounded domain $\DD$ converges.
When the fluid is a horizontal layer and the gravity field is
constant in the negative $z$ direction, we take $U_+=(z-1)_+,$ where
$z_+$ denotes the function given by $z$ when $z>0$ and by 0 when
$z<0$. The factor $g$ has been absorbed into the pure parameter
$\lambda$.
\end{remark}


\section{Free boundary flows with vorticity.}

Free boundary value flows with vorticity, with both gravitational
forces and surface tension included, are generated by the
Hamiltonian
\begin{equation}\label{H1}
H=\iiint\limits_\DD\EE\,d^3\bfx+ \sigma\iint\limits_\Sigma dS,
\qquad \EE=\frac{\bfv\cdot\bfv}2+\lambda\, U_+(\bfx)
\end{equation}
The corresponding Poisson brackets are \cite{LMMR86}
\begin{multline}
\{F,G\,\}=\iiint\limits_\DD \pdl{ F}{ \bfv}\cdot
\left(\curl\,\bfv\times\pdl{
G}{\bfv}\right)\,d^3\bfx\\[4mm]+\iint\limits_{\Sigma} \left(\pdl{ F}{\varphi}\pdl{
G}{\Sigma}-\pdl{ F}{\Sigma}\pdl{ G}{\varphi}\right)dS,\label{PB}
\end{multline}
where $\Sigma$ is the free boundary and $dS$ is the element of
surface area on $\Sigma$.

Admissible functionals are regarded as functions of $\bfv$ and
$\Sigma$, the free boundary of $\DD$, and their gradients are
defined implicitly by the relation
\begin{equation*}
\frac{d}{d\varepsilon}F(\bfv_\varepsilon,\Sigma_\varepsilon)\Big\vert_{\varepsilon=0}=
\iiint\limits_\DD \pdl{F}{\bfv}\cdot
\delta\bfv\,d^3\bfx+\iint\limits_{\Sigma}\pdl{F}{\Sigma}\delta\Sigma\,dS.
\end{equation*}
Variations with respect to the free surface are restricted to normal
variations, in a sense explained below. Admissible functionals $F$
are those for which $\delta F/\delta \bfv$ is a divergence free
vector field. We require that $\iint_\DD \delta\Sigma dS=0$,
reflecting the fact that only volume preserving variations are
allowed. This means that the gradient of a functional with respect
to $\Sigma$ is determined only up to a constant.

Let $\LL_d(\DD)$ be the space of divergence free $L^2$ vector fields
on $\DD.$ Let $P_1$ and $P_2$ be the orthogonal projections defined
by $P_1\bfv=\bfw$ and $P_2\bfv=\grad\varphi$ in the Weyl-Hodge
decomposition.

\begin{lemma}\label{grads}
Let
\begin{equation*}
F(\bfv,\Sigma)=\iiint\limits_\DD
\FF(\bfv,\bfx)d^3\bfx+\sigma\iint\limits_\Sigma dS
\end{equation*}
be an admissible functional. Then
\begin{equation*}
\pdl{F}{\bfw}=P_1\pdl{F}{\bfv}\in \LL^2(\DD,d^3\bfx).
\end{equation*}
The gradients with respect to $\varphi$ and $\Sigma$ lie in
$\LL^2(\Sigma,dS)$ and are given by
\begin{equation*}
\pdl{F}{\varphi}=\pdl{F}{\bfv}\Big\vert_\Sigma\cdot \nu, \qquad
\pdl{F}{\Sigma}=\FF(\bfv,\bfx)+\sigma\kappa\Big\vert_\Sigma\ \
\text{mod constant},
\end{equation*}
where $\kappa$ is the mean curvature function on $\Sigma.$
\end{lemma}

\begin{proof}
Applying the Weyl-Hodge decomposition to both $\delta\bfv$ and
$\delta F/\delta \bfv$ we obtain
\begin{equation*}
\left\langle \pdl{F}{\bfv},\delta\bfv\right\rangle=\iiint\limits_\DD
\pdl{F}{\bfv}\cdot\delta\bfv\,d^3\bfx=\iiint\limits_\DD
P_1\pdl{F}{\bfv}\cdot\delta\bfw+P_2\pdl{F}{\bfv}\cdot\delta\grad\varphi\,d^3\bfx.
\end{equation*}
By the uniqueness of the Weyl-Hodge decomposition, we may conclude
\begin{equation*}
\pdl{F}{\bfw}=P_1\pdl{F}{\bfv}, \qquad
\pdl{F}{\varphi}=P_2\pdl{F}{\bfv}.
\end{equation*}
Since $\delta F/\delta\bfv$ is divergence free, we have,
by the divergence theorem,
\begin{equation*}
\iiint\limits_\DD
P_2\pdl{F}{\bfv}\cdot\delta\grad\varphi\,d^3\bfx=\iiint\limits_\DD
\pdl{F}{\bfv}\cdot\grad\delta\varphi\,d^3\bfx=\iint\limits_{\partial\DD}
\pdl{F}{\bfv}\cdot \nu \,\delta\varphi\,dS,
\end{equation*}
and the second relation follows.

Let $\Sigma_\varepsilon$ be a one parameter family of surfaces
parameterized by a vector valued map
\begin{equation*}
\bfX(u,v,\varepsilon)=\bfX_0(u,v)+\varepsilon\,\delta\Sigma\,\bfN(u,v),
\end{equation*}
where $\bfN$ is the normal vector field to $\Sigma$. For
$\varepsilon$ sufficiently small, the symmetric difference
$\DD_\varepsilon\Delta\DD$ of the domains bounded respectively by
$\Sigma_{\varepsilon}$ and $\Sigma$ is contained in a tubular
neighborhood of $\Sigma$. In this neighborhood, the volume element
of the $3$-space can be written as $ d^{3}\bfx = dr dS$ where $dS$
is the area element on $\Sigma$ and $dr$ corresponds to the normal
coordinate in the tubular neighborhood. We get
\begin{align*}
\delta\iiint\limits_\DD \FF(\bfv,\bfx)\,d^3\bfx  =&
\lim_{\varepsilon\to 0}\frac{1}{\varepsilon}
\iiint\limits_{\DD_\varepsilon\Delta\DD}\FF(\bfv,\bfx)\,d^3\bfx\\[4mm]
 =& \iint\limits_{\Sigma}\left( \lim_{\varepsilon\to 0}\frac{1}{\varepsilon}
\int_0^{\varepsilon\delta\Sigma}\FF(\bfv,\bfx)\, dr\right)dS  =
\iint\limits_{\Sigma} \FF(\bfv,\bfx) \delta\Sigma \, dS .
\end{align*}
On the other hand, by classical differential geometry,
\begin{equation*}
\delta\iint\limits_\Sigma dS=\iint\limits_\Sigma \kappa \delta\Sigma
\, dS,
\end{equation*}
where $\kappa$ is the mean curvature function on $\Sigma.$ This
completes the proof of Lemma \ref{grads}.
\end{proof}

Let us derive the equations of motion from the Hamiltonian
structure. We have
\begin{equation*}
\pdl{H}{\bfv}=\bfv, \qquad \pdl{H}{\varphi}=\bfv\cdot \nu, \qquad
\pdl{H}{\Sigma}=\EE\Big\vert_\Sigma+\sigma \kappa.
\end{equation*}
>From $\dot F=\{H,F\,\}$, we get
\begin{multline}\label{dotF=}
\iiint\limits_\DD \pdl{F}{\bfv}\cdot
\bfv_t\,d^3\bfx+\iint\limits_\Sigma
\pdl{F}{\Sigma}\Sigma_t\,dS=\\[4mm]
\iiint\limits_\DD (\bfv\times
(\curl\,\bfv))\cdot\pdl{F}{\bfv}\,d^3\bfx+\iint\limits_\Sigma
\left(\bfv\cdot\nu\pdl{F}{\Sigma}
-(\EE+\sigma\kappa)\pdl{F}{\varphi}\right)\,dS.
\end{multline}

Since
\begin{equation*}
\iint\limits_\Sigma \EE\pdl{F}{\varphi}\,dS= \iint\limits_\Sigma
\EE\pdl{F}{\bfv}\cdot \nu\,dS =\iiint\limits_\DD
\grad\EE\cdot\pdl{F}{\bfv}\,d^3\bfx
\end{equation*}
and $\delta F/\delta \bfv$ is divergence free, we get from
\eqref{dotF=}, using functionals for which $\delta
F/\delta\varphi=0$,
\begin{equation*}
\bfv_t+(\curl\,\bfv)\times\bfv=\,\grad\,(-p+\EE), \qquad
\Sigma_t=\bfv\cdot\nu\Big\vert_\Sigma.
\end{equation*}
The boundary condition on the bottom is $\bfv \cdot \nu=0$, where
$\nu$ is the outward normal.

The first equation, together with \eqref{vecid} imply that $\Delta
p=-\divg\,(\bfv\cdot\grad)\bfv.$ Substituting the two equations
above into \eqref{dotF=}, we obtain
\begin{equation*}\iiint\limits_\DD
\grad\,p\cdot\pdl{F}{\bfv}\,d^3\bfx-\iint\limits_\Sigma
\sigma\kappa\pdl{F}{\varphi}\,dS=0
\end{equation*}
for all admissible functionals $F$. Applying the divergence theorem
to the integral over $\DD$ we obtain
\begin{equation*}
\iint\limits_\Sigma (p-\sigma\kappa)\pdl{F}{\varphi}\,dS=0,
\end{equation*}
for all admissible functionals $F$. But
\begin{equation*}
\iint\limits_\Sigma \pdl{F}{\varphi}\,dS= \iint\limits_\Sigma
\pdl{F}{\bfv}\cdot\nu\,dS=\iiint\limits_\DD
\divg\,\pdl{F}{\bfv}\,d^3\bfx=0;
\end{equation*}
and therefore
\begin{equation}\label{p=o}
p\Big\vert_\Sigma=\sigma\kappa+constant.
\end{equation}
Thus the Hamiltonian approach yields the dynamic conditions on the
free boundary in the case of surface tension.
\cite{ASS06,LMMR86,Joh97}.

\begin{remark}
In the general theory one considers normal variations of the free
surface, whereas in the theory of gravity waves on a free surface
over a horizontal bottom, it is customary to use the height of the
free surface, $\zeta$. More generally, if the surface $\Sigma$ is a
graph over a fixed manifold $\MM$, we may represent $\Sigma$ by a
``height'' function $\zeta$ defined on $\MM$. In that case we refer
to $\delta\zeta$ as the ``vertical'' variation and $\delta\Sigma$ as
the ``normal'' variation.
\end{remark}

\begin{proposition}\label{graph}
Let $\delta\Sigma$ and $\delta\zeta$ denote the normal and vertical
variations of a surface $\Sigma$ in the case when $\Sigma$ is a
graph over a fixed manifold. Let $\Sigma$ be given in local
coordinates by $\phi=0$, where $\phi=z-\zeta.$ Then
$\delta\zeta=|\grad\phi|\,\delta\Sigma.$
\end{proposition}

\begin{proof}
Let $\bfX:U\mapsto \R^3$ be a local embedding of $\Sigma$ in $\R^3$;
and let $\bfX_\varepsilon$ be a one parameter family of embeddings,
with $\bfX_0=\bfX$. Then
\begin{equation*}
\delta\Sigma=\frac{d}{d\varepsilon}\Big\vert_{\varepsilon=0}
\left(\bfX_\varepsilon (u,v)\right)\cdot\nu.
\end{equation*}
Let $\Sigma$ be defined by $\phi=0,$ $\phi=z-\zeta$. Then
$\bfX_\varepsilon=(u,v,\zeta_\varepsilon(u,v);$ and
\begin{equation*}
\delta\Sigma=\frac{d}{d\varepsilon}\Big\vert_{\varepsilon=0}
\begin{pmatrix}u\\v\\
\zeta_\varepsilon(u,v)
\end{pmatrix}\cdot\nu=
\begin{pmatrix}0\\0\\
\delta\zeta
\end{pmatrix}\cdot
\frac{\grad\phi}{|\grad\phi|}=\frac{\delta\zeta}{|\grad\phi|}.
\end{equation*}
\end{proof}


\section{Variational principles for traveling waves}

The Hamiltonian structure of the equations for gravity waves can be
used to obtain variational principles for traveling waves -- waves
of constant speed and shape. Such a wave is a stationary solution of
the Hamiltonian system in a Galilean frame moving with the wave;
thus the wave is a critical point for the Hamiltonian, computed in
such a reference frame. We apply the method here to the general case
of gravity waves on a horizontal surface.  The variational principle
for irrotational flows given below appears to be new.

A variational approach, if successful, would  permit a global
treatment of the existence of traveling waves by the direct methods
of the calculus of variation; but so far,  the existence of
traveling waves for potential flows of low amplitude have been
proved by perturbation methods. The first existence theorems were
given independently for periodic wave trains by Levi-Civita
\cite{Lev25}  and  Struik \cite{Str26} in the case of finite depth.
The existence of the solitary wave, which is a more difficult
problem, was first proved by Friedrichs and Hyers \cite{FH54}, since
the bifurcation problem in this case is a singular perturbation
problem (see the discussion by Sattinger \cite{DHS06}). These
authors used conformal mapping techniques. A dynamical systems
approach to the existence of traveling waves has been developed by
Kirchg\"assner \cite{Kir88}; Amick and Toland  \cite{AT81} have
shown that periodic wave trains tend to a solitary wave in the limit
as the period tends to infinity.

In the direct method, one first uses compactness properties of the
functional to obtain a minimum from a minimizing sequence. In
general, this guarantees only a weak solution of the associated
Euler-Lagrange equations. In many cases, these are elliptic
equations, and it is possible to prove sufficient regularity of the
weak solution to show that in fact it is a classical solution to the
problem. (See Alt and Caffarelli \cite{AC81} for functionals of
Friedrichs' type.) For the present, we simply indicate the method
for the problems discussed here in the theorems below.

\begin{theorem}\label{var1}
Euler's equations for gravity waves are the Euler-Lagrange equations
for the functional
\begin{equation}\label{H}
    \HH(\varphi,\zeta) = \iint\limits_{\DD_\zeta} \left[ \frac12 \left[ (\grad\varphi)^2-1 \right] + \lambda (y-1)_{+} \right] \,d^2\bfx \, ,
\end{equation}
where
\begin{equation*}
    y_{+} = \left\{
              \begin{array}{ll}
                0, & \hbox{$y \le 0$;} \\
                y, & \hbox{$y \ge 0$.}
              \end{array}
            \right.
\end{equation*}
\begin{equation*}
    \DD_\zeta = \set{ (x,y)\,:\,-\infty<x<\infty,\ 0\le y\le 1+\zeta(x) };
\end{equation*}
and the minimum is taken over all functions $\varphi$ for which
\begin{equation*}
    \iint\limits_\DD \left[ (\varphi_x-1)^2+\varphi_y^2 \right] \, d^2\bfx < +\infty .
\end{equation*}
If $(\varphi,\zeta)$ is a local minimum of $\HH$, then $\varphi$ is
harmonic on the interior of $\DD_{\zeta}$; if $\zeta$ is $C^1$ and
$\varphi\in H^2(\DD)$, then the kinematic and Bernoulli equations
hold on the free surface.
\end{theorem}

\begin{remark}
The Hamiltonian \eqref{H} is the renormalization of the Hamiltonian
in the moving frame. By carrying out the integration in $y$ we
obtain
\begin{equation*}
\iint\limits_\DD (y-1)_+\,d^2\bfx=\frac12\int_{-\infty}^\infty
\zeta^2\,dx;
\end{equation*}
thus $\HH$ can also be written
\begin{equation*}
\HH(\varphi,\zeta)=\iint\limits_\DD \frac12
\left[(\grad\varphi)^2-1\right]
\,d^2\bfx+\frac\lambda2\int_{-\infty}^\infty \zeta^2(x)\,dx.
\end{equation*}
If $\varphi\in H^2(\DD)$ and $\zeta\in C^1$, then $\grad\varphi$ has
an $L^2$ trace on the boundary $y=\zeta$, and Stokes theorem
applies.
\end{remark}

\begin{proof}
Let $(\varphi,\zeta)$ be a minimizer of $\HH$ and suppose that
$\zeta$ is $C^1$ and $\varphi\in H^2(\DD)$.  Let
$(\varphi_\varepsilon,\zeta_\varepsilon)$ be a one parameter family
of admissible functions and denote the corresponding domains by
$\DD_\varepsilon$. By the calculations in \S\ref{zak} we have
\begin{align*}
\delta \HH(\delta\varphi,\delta\zeta)
    & = \pd{H(\varphi_\varepsilon, \zeta_\varepsilon)}{\varepsilon}\Big\vert_{\varepsilon=0} \\
    & = - \iint\limits_{\DD} \Delta\varphi\delta\varphi \, d^2 \bfx
      + \oint\limits_{\partial \DD} \varphi_\nu\delta\varphi\,ds
      + \inmp \left[\frac12(\grad\varphi)^2-\frac12+\lambda\zeta\right] \delta\zeta\,dx \\
    & = 0,
\end{align*}
for all admissible $\delta\varphi$, $\delta \zeta$.

Since the bottom is fixed, $\varphi_\nu=0$ on $y=0$. We first
restrict ourselves to variations for which
$\delta\zeta=\delta\varphi\big\vert_S=0$. Then the double integral
must vanish for a set of variations $\delta\varphi$ which are dense
in $L^2(\DD)$; it follows that $\varphi$ is harmonic in the interior
of $\DD$. As before, $\varphi_\nu ds=\grad\varphi\cdot
(-\zeta_x,1)dx=(\varphi_y-\zeta_x\varphi_x)dx;$ and so
\begin{equation*}
\delta \HH=\inmp
\left[\left(\frac{(\grad\varphi)^2-1}2+\lambda\zeta\right)
\delta\zeta+(\varphi_y-\varphi_x\zeta_x)\,\delta\varphi\right]dx.
\end{equation*}
Setting first $\delta\zeta=0$ and letting $\delta\varphi$ vary on
$\Sigma$, we obtain the kinematic equation on the free surface.
Therefore the second term always vanishes. Now allowing
$\delta\zeta$ to vary, we see that Bernoulli's equation holds on
$\Sigma$.
\end{proof}


\section{A variational problem with constraint}

Whereas Friedrich's paper shows that Bernoulli's equation is not
obtained when the functional $J$ is minimized with respect to the
stream function, Constantin et. al. showed in \cite{ASS06} that
traveling gravity waves in the rotational case are obtained as
extremals of a variational problem for the stream function with
constraints. The existence of traveling water waves with vorticity
was established in \cite{CS04} for the periodic case. In a recent PhD thesis at Brown University, V. Hur
\cite{Hur} has constructed solitary waves with non-zero vorticity.  Some of
their qualitative properties were investigated in \cite{CE04}.

In the irrotational case we have
\begin{theorem}\label{constrained}
Define the set of admissible functions $\KK=\{\psi,\zeta\}$ with the
following properties
\begin{align*}
i) & \qquad\inmp\zeta(x)dx=m;\ \inmp \zeta^2\,dx<\infty;\\[4mm]
ii) & \qquad \psi(x,0)=0;\ \psi(x,1+\zeta(x))=1,\\[4mm]
iii) & \qquad \iint\limits_\DD
\left[\psi_x^2+(\psi_y-1)^2\right]d^2\bfx <+\infty.
\end{align*}
Consider the variational problem
\begin{equation*}
\lambda = \inf_{\KK} \frac{\iint_{\DD} \left[ (\grad\psi)^2-1
\right] d^2\bfx}{\inmp\zeta^2\,dx}.
\end{equation*}
Let $(\psi,\zeta)$ be a minimizer in $\KK$ of the above variational
principle. Then $\psi$ is harmonic in the interior of $\DD$. If
$\zeta$ is $C^1$, and $\psi\in H^2(\DD)$, then the Bernoulli
equation is satisfied on the free surface $\psi=1$. Hence minima of
the above variational problem provide an irrotational flow for the
gravity wave problem.
\end{theorem}

\begin{proof}
Let $(\psi,\zeta)$ be a minimizer, and let
$\psi_\varepsilon,\zeta_\varepsilon$ be a family of admissible
functions with $\psi_0=\psi$ and $\zeta_0=\zeta$. Then
$J(\varepsilon)\ge 0$ and $J(0)=0$, where
\begin{equation*}
J(\varepsilon) =  \iint\limits_{\DD_\varepsilon} \left[
(\grad\psi_\varepsilon)^{2} - 1\right]
\,d^2\bfx-\lambda\inmp\zeta_\varepsilon^2\,dx.
\end{equation*}
Then $\delta J(\delta\psi,\delta\zeta)=0$ for all admissible
variations, where
\begin{align*}
\delta J & = \iint\limits_{\DD} 2\grad\psi\cdot\grad\,\delta\psi \, d^2\bfx + \inmp \left[ (\grad\psi)^2-1-2\lambda\zeta \right]\delta\zeta\,dx \\
         & = -2\iint\limits_{\DD} \Delta\psi\,\delta\psi\,d^2\bfx +
2\oint\limits_{\partial\DD}\delta\psi\psi_\nu \,ds + \inmp \left[
(\grad\psi)^2-1-2\lambda\zeta \right] \delta\zeta \, dx.
\end{align*}

The integral over the bottom of the flow domain vanishes, since
$\psi_\nu=0$ there. On the free surface (see $\delta(2)$, p. 65
\cite{Fri33})
\begin{equation*}
\delta\psi+\psi_y\delta\zeta=0.
\end{equation*}
This follows immediately by differentiating the relation
$\psi_\varepsilon(x,1+\zeta_\varepsilon(x))\equiv 1$ with respect to
$\varepsilon$ and setting $\varepsilon$ equal to zero. Similarly,
differentiating the expression $\psi(x,1+\zeta(x))\equiv 1$ with
respect to $x$ we find that $\psi_x/\psi_y=-\zeta_x$; hence
\begin{equation*}
\psi_\nu=\grad\psi\cdot{\bf\nu}=\grad\psi\cdot\frac{\grad\psi}{||\grad\psi||}
=\frac{(\grad\psi)^2}{\sqrt{\psi_x^2+\psi_y^2}}=\frac{(\grad\psi)^2}{|\psi_y|\sqrt{1+\zeta_x^2}}.
\end{equation*}

Hence $\delta J$ reduces to
\begin{equation}\label{var}
\delta J = -2\iint\limits_{\DD} \Delta\psi \,\delta\psi \,d^2\bfx -
\inmp \left[ (\grad\psi)^2+1+2\lambda\zeta \right] \delta\zeta\,dx.
\end{equation}
First restrict the variations to fixed domains, $\delta\zeta=0$, and
the first integral must vanish for all variations $\delta\psi$ which
vanish on $\partial\DD$. Hence $\psi$ is harmonic in the interior of
$\DD$, and the double integral vanishes.

We next consider variations of the domain. Since $\int
\zeta_\varepsilon\,dx= m$, for all variations, we have
$\int\delta\zeta\,dx=0$; then the condition
\begin{equation*}
\inmp((\grad\psi)^2+1+2\lambda\zeta)\delta\zeta\,dx=0
\end{equation*}
for all such $\delta\zeta$ implies that the integrand is a constant. We therefore have
$(\grad\psi)^2+2\lambda\zeta+1=C=const.$ on the line; letting
$x\to\infty$ and noting that $\zeta\to 0$ while $(\grad\psi)^2\to
1$ we see that $C=2$, and the Bernoulli equation is satisfied.
\end{proof}


\paragraph*{Acknowledgement} The results in this paper were obtained
during the authors' visit to the Mittag-Leffler Institute in
October, 2005, in conjunction with the Program on Wave Motion. The
authors wish to extend their thanks to the Institute for its
generous sponsorship of the program, as well as to the organizers
for their work.


\end{document}